\documentclass[11pt]{llncs}

\usepackage{color}
\usepackage{url}

\usepackage[usenames,dvipsnames]{xcolor}
\usepackage{multirow}
\usepackage{graphics}
\usepackage{a4wide}
\usepackage{pdfpages}

\newcommand{\cT}{\mathcal{T}}

\newcommand{\blue}[1]{\textcolor{black}{#1}}

\newcommand{\red}[1]{\textcolor{black}{#1}}
\newcommand{\magenta}[1]{\textcolor{black}{#1}}

\newcommand{\revision}[1]{\textcolor{black}{#1}}

\usepackage{amsmath}
\usepackage{graphicx}

\title{Deciding the existence of a cherry-picking sequence is hard on two trees}
\author{Janosch D\"ocker\inst{1}, Leo van Iersel\inst{2}, Steven Kelk\inst{3}, Simone Linz\inst{4}}

\institute{Department of Computer Science, University of T\"ubingen, Germany\\
\email{janosch.doecker@uni-tuebingen.de}
\and
Delft Institute of Applied Mathematics, Delft University of Technology\\
PO-box 5031, 2600 GA Delft, The Netherlands\\
\email{L.J.J.vanIersel@tudelft.nl}
\and Department of Data Science and Knowledge Engineering (DKE),\\ Maastricht University, The Netherlands,\\ \email{steven.kelk@maastrichtuniversity.nl}
\and Department of Computer Science, University of Auckland, New Zealand,\\
\email{s.linz@auckland.ac.nz}}

\providecommand{\keywords}[1]{\textit{Keywords:} #1}

\begin{document}
\maketitle

\begin{abstract}
Here we show that deciding whether two rooted binary phylogenetic trees on the same set of taxa permit a \emph{cherry-picking sequence}, a special type of elimination order on the taxa, is \red{NP-complete}. This improves on an earlier result which proved hardness for eight or more trees. Via a known equivalence between cherry-picking sequences and temporal phylogenetic networks,  our result proves that 
it is \red{NP-complete} to determine
the existence of a 
temporal phylogenetic network that contains
topological embeddings of both trees. The hardness result also
greatly strengthens previous
inapproximability results for
the minimum temporal-hybridization number problem. 
This is the optimization version of the problem where we wish to construct a temporal phylogenetic network that topologically embeds two given rooted binary phylogenetic trees and that has a minimum number of indegree-2 nodes, which represent events such as hybridization and horizontal gene transfer.
We end on a positive note, pointing out that fixed parameter tractability results in this area are likely to ensure the continued relevance of the temporal phylogenetic network model.
\end{abstract}


\section{Introduction}

In the field of phylogenetics it is common to represent the evolution of a set of species $X$ by a rooted phylogenetic tree; essentially a rooted, bifurcating tree whose leaves are bijectively labeled by $X$ \cite{SempleSteel2003}. Driven by the realization that evolution is not always treelike there has been growing attention for the construction of phylogenetic \emph{networks}, which generalize phylogenetic trees to directed acyclic graphs \cite{Bapteste13,Gusfield14,HusonRuppScornavacca10,Soucy15}. 
One well known optimization problem for phylogenetic networks is as follows: given a set of rooted phylogenetic trees $\mathcal{T}$ on the same set of taxa $X$, compute a phylogenetic network $N = (V,E)$ which \emph{displays} (i.e. contains topological embeddings of) all the trees in $\mathcal{T}$, such that the \emph{reticulation number} $|E|-(|V|-1)$ is minimized. When $N$ is restricted to being \emph{binary} this is equivalent to minimizing the number of nodes of $N$ with indegree-2. This optimization model is known as \emph{minimum hybridization} and it has been extensively studied in the last decade (see e.g. \cite{sempbordfpt2007,chen2010hybridnet,approximationHN,van2016kernelizations,whidden2013fixed}). More recently variations of minimum hybridization have been proposed which constrain the topology of $N$ to be more biologically relevant. One such constraint is to demand that $N$ is \emph{temporal}~\cite{temporal}. Informally, a phylogenetic network $N$ is temporal if (i) the nodes of $N$ can be labeled with times, such that nodes of indegree-2 have contemporaneous parents, and time moves strictly forwards along treelike parts of the network; and (ii) each non-leaf vertex has a child whose indegree is 1. Property (ii) by itself is referred to as {\it tree-child} in the literature~\cite{tree-child}. It has been shown that when $|\mathcal{T}|=2$ it is NP-hard to solve the minimum temporal-hybridization number problem to optimality \cite{humphries2013complexity}. To establish the result, the authors proved that the problem is in fact APX-hard, which implies that for some constant $c > 1$ it is not possible in polynomial time to approximate the optimum within a factor of $c$, unless
P=NP \cite{papayanna}.

A more fundamental question remained, however, open: is it  possible in polynomial time to determine if \emph{any} temporal phylogenetic network exists that displays the input trees, regardless of how large $|E|-(|V|-1)$ is \cite{humphries2013cherry,steel2016phylogeny}? Here we settle
this question by showing that, even for $|\mathcal{T}|=2$, it is \red{NP-complete} to determine whether such a network exists. We prove this by using the \emph{cherry-picking} characterization of temporal phylogenetic networks introduced in \cite{humphries2013cherry}. There it was shown that, given an arbitrarily large set $\mathcal{T}$ of rooted binary phylogenetic trees on $X$, there exists a temporal phylogenetic network that displays each tree in $\cT$ precisely if $\cT$ has a so-called cherry-picking sequence. Informally, a \emph{cherry-picking sequence} on $\mathcal{T}$ is an elimination order on $X$ that deletes one element of $X$ at a time, where at each step only elements can be deleted which are in a cherry of every tree in $\cT$ \cite{humphries2013cherry}. We show here that the seminal \revision{NP-complete} problem 3-SAT \cite{Karp1972} can be reduced to the question of whether two trees permit a cherry-picking sequence. This improves upon a recent result by two of the present authors which shows that, for $|\mathcal{T}| \geq 8$, it is NP-complete to determine whether $\mathcal{T}$ has a cherry-picking sequence \cite{8trees}. \red{Our} hardness result is highly non-trivial and requires extensive gadgetry; \red{to clarify we include an explicit example of the construction after the main proof.}

As we discuss in the final section of the paper, this result has quite significant negative consequences: given that the decision problem is already hard, the
\red{minimum temporal-hybridization number} problem is in some sense ``effectively inapproximable'',
even for two trees. This greatly strengthens the earlier APX-hard inapproximability result. \red{Nevertheless, as we subsequently point out, positive fixed parameter tractability (FPT) \cite{cygan2015parameterized} results for the minimum temporal-hybridization number problem do already exist \cite{humphries2013cherry} and our results emphasize the importance of further developing such algorithms, since fixed parameter tractability forms the most promising remaining avenue towards practical exact methods.}


\section{Preliminaries}

A \emph{rooted binary phylogenetic tree} on a set of taxa $X$, where $|X| \geq 2$, is
a rooted, connected, directed tree with a unique \emph{root} (a vertex of indegree-0
and outdegree-2), where the leaves (vertices with indegree-1 and outdegree-0) are bijectively labeled by $X$, and where all interior vertices of the tree are indegree-1 and outdegree-2. If $|X|=1$, we consider the single isolated node labeled by the unique element of $X$, to also be a rooted binary phylogenetic tree. Since all phylogenetic trees considered in this paper are rooted and binary, we henceforth write \emph{tree} for brevity, and draw no distinction between the elements of $X$ and the leaves they label. Let $T$ be a tree, and let $\cT=\{T_1,T_2,\ldots,T_m\}$ be a set of trees. We use $X(T)$ to denote the taxa set of $T$ and, similarly, we use $X(\cT)$ to denote the union of taxa sets over all elements in $\cT$, i.e. $X(\cT)=X(T_1)\cup X(T_2)\cup\ldots\cup X(T_m)$. Lastly, for two distinct elements $x$ and $y$ in $X$, we call $\{x,y\}$ a \emph{cherry} of $T$ if they have the same parent. A tree with a single cherry is referred to as \red{a} {\it caterpillar}.

Now, let $T$ be a tree on $X$, and let $X'=\{x_1,x_2,\ldots,x_k\}$ be an arbitrary set. We write $T|X'$ to denote the
tree obtained from $T$ by taking the minimum subtree spanning the elements of $X'$
and repeatedly suppressing all vertices with indegree-1 and outdegree-1. \revision{(If $v$ is a vertex with indegree-1 and outdegree-1, with incident edges $(u,v)$ and $(v,w)$, then \emph{suppressing} $v$ is achieved as follows: $v$ and its two incident edges are deleted, and an edge $(u,w)$ is added.)}

Furthermore, we also write $T[-x_1,x_2,\ldots,x_k]$ or $T[-X']$ for short to denote $T|(X-X')$. If $X\cap X'=\emptyset$, then $T|X'$ is the null tree  and $T[-X']$ is $T$ itself.
For a set $\cT=\{T_1,T_2,\ldots,T_m\}$ of trees on subsets of $X$, we write $\cT|X'$ (resp. $\cT[-X']$) when referring to the set  $\{T_1|X',T_2|X',\ldots,T_m|X'\}$ (resp.~$\{T_1[-X'],T_2[-X'],\ldots,T_m[-X']\}$).  Lastly, a rooted binary phylogenetic tree is {\it pendant} in $T$ if it can be detached from $T$ by deleting a single edge.

\subsection{Cherry-picking sequence problem on trees with the same set of taxa}
\label{subsec:classical}

\begin{figure}[t]
\centering
\includegraphics[]{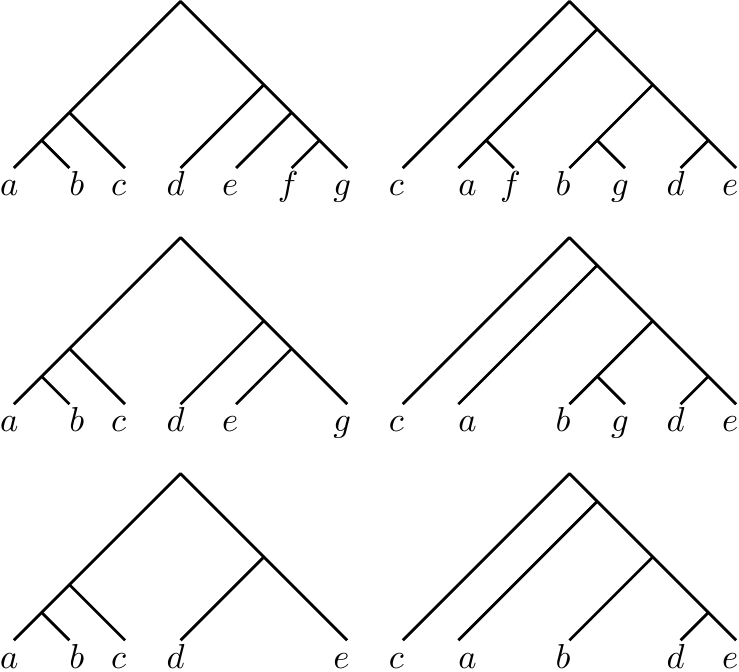}
\caption{A cherry-picking sequence for the two trees $T$ and $T'$ at the top is $(f,g,d,b,a,c,e)$. The two trees in the middle have been obtained from $T$ and $T'$, respectively, by pruning $f$, and the two trees at the bottom have been obtained from $T$ and $T'$ by first pruning $f$ and, subsequently, pruning $g$. While we can alternatively prune $a$ and, subsequently, $b$, from $T$ and $T'$, note that no cherry-picking sequence exists for $T$ and $T'$ whose first two elements are $a$ and $b$.}
\label{fig:CPSexample}
\end{figure}

We say that
a taxon $x \in X$ is \emph{in} a cherry of $T$ if there exists
some $y \neq x$ such
that $\{x,y\}$ is a cherry of $T$ or $T$ consists of a single leaf $x$. If $x$ is in a cherry of $T$, we say that
$x$ is \emph{picked} (or \emph{pruned}) from $T$ to denote the operation of
replacing $T$ with $T[-x]$. 
Given a set
of trees $\mathcal{T}$, all on the same set of taxa $X$, we say that a taxon $x \in X$
is \emph{available (for picking)} in $\mathcal{T}$ if $x$ is in a cherry in each tree in $\mathcal{T}$. When this is the case, we say that $x$ is \emph{picked} (or \emph{pruned}) from $\mathcal{T}$ to denote the operation of replacing $\cT$ with $\cT[-x]$. 

Let $\mathcal{T}$ be a set of trees on the same set of taxa $X$. A \emph{cherry-picking sequence} is an order on $X$, say $(x_1, x_2, \ldots, x_{|X|})$, such that each $x_i$ with $i\in\{1,2,\ldots,|X|\}$ is available in $\cT[-x_1,x_2,\ldots,x_{i-1}]$.
Such a sequence is not guaranteed to exist; if it does, we say that $\mathcal{T}$ \emph{permits} a cherry-picking sequence. It was shown in \cite{8trees} that \red{deciding} whether such a sequence exists is NP-complete if $|\mathcal{T}| \geq 8$. Note that, if $|\cT|=1$, then $\cT$ always has a cherry-picking sequence. To illustrate, a cherry-picking sequence for the two trees that are shown at the top of Figure~\ref{fig:CPSexample} is $(f,g,d,b,a,c,e)$.

\subsection{A more general cherry-picking sequence problem}

Let $\cT$ be a set of trees, and let $X = X(\cT)$. Suppose we consider the variant of the problem described in Section~\ref{subsec:classical} in which the trees in $\mathcal{T}$ do not necessarily have the same set of taxa. In this case, some taxa may be missing from some trees. This requires us to
generalize the concept of being \emph{in} a cherry of a tree. 
We say that a taxon $x$ is {\it in} a cherry of a tree $T$, if exactly
one of the following conditions holds:

\begin{enumerate}
\item $x \not \in X(T)$ or
\item $x \in X(T)$ and \revision{$T$ has a cherry $\{x,y\}$, where $x$ and $y$ are distinct elements in $X(T)$.}
\end{enumerate}
\red{(Note that, once again, this means that if $x$ is the only taxon
in $T$, then $x$ is vacuously considered to be in a cherry of $T$.)} It initially seems counter-intuitive to say, when condition 1 
applies, that $x$ is ``in'' a cherry of $T$. However, the idea behind this is that such trees do not constrain whether $x$ can be picked; they ``do not care''. More formally, we say that a taxon $x$ is \emph{available} in $\mathcal{T}$ if it is in a cherry in each tree in $\mathcal{T}$. Similar to Section~\ref{subsec:classical}, we say that an order on $X$, say $(x_1,x_2 \ldots, x_{|X|})$ is a {\it cherry-picking sequence} of $\mathcal{T}$ if each $x_i$ with $i\in\{1,2,\ldots,|X|\}$ is available in $\cT[-x_1,x_2,\ldots,x_{i-1}]$.
If a tree \red{becomes the null tree due to all its taxa being pruned away} then this tree plays no further role.
Moreover we note that, if all trees in $\mathcal{T}$ have the same set of taxa, then the more general definition of a cherry-picking sequence given in this subsection and that will be used throughout the rest of this paper coincides with that given in Section \ref{subsec:classical}.

\section{Main results}
In this section, we establish the main result of this paper. We start with two lemmas.

\begin{lemma}
\label{lem:sametaxa}
Let $\mathcal{T}$ be a set of $m$ trees on not necessarily the same set of taxa.
Then we can construct in polynomial time a set $\mathcal{T}'$ of $m$ trees all on the same set of taxa, such that $\mathcal{T}$ has a cherry-picking sequence if and only if $\mathcal{T}'$ does.
\end{lemma}
\begin{proof} Let $X=X(\cT)$, and let $Y = \{ y_1, y_2, \ldots\}$ be the set of taxa that are missing from at least one input tree. Let $Y' = \{y'_1, y'_2, ...\}$ be a disjoint copy of this set. Every modified tree will have taxon set $X \cup Y' \cup \{\rho\}$.  The idea is as follows. Let $T_{Y'}$ be an arbitary rooted binary tree on $Y'$. For each input tree $T_i$, we start by joining $T_i$ and  $\rho$ beneath a root, and then join this new tree and $T_{Y'}$ together beneath a root. Next, for each $y_j \in Y$ that is missing from $T_i$, we add $y_j$ by subdividing the edge that feeds into $y'_j$ and
attaching $y_j$ there (so $y_j$ and $y'_j$ become siblings). For an example, see Figure \ref{fig:sameSetOfTaxa}. We call the set of trees constructed in this way $\mathcal{T}'$. 
The high-level idea is that if a tree $T_i$ does not contain some taxon $x$, we attach $x$ just above $x'$ and thus ensure that, trivially, $x$ is in a cherry in that tree (i.e. together with $x'$). So $T_i$ does ``not care'' about $x$ and will not obstruct it \red{from} being pruned.

\begin{figure}
\centering
\includegraphics[]{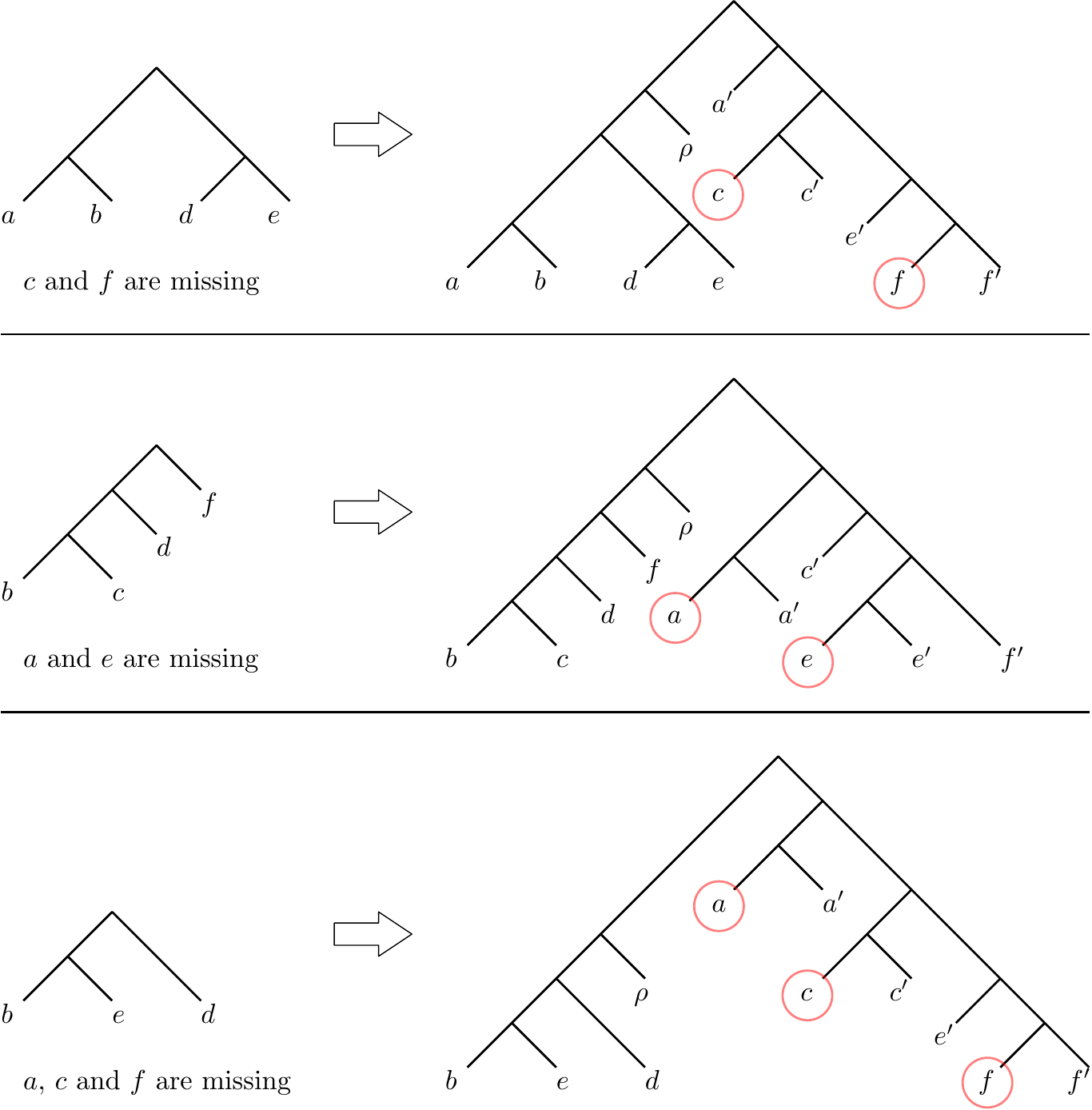}
\caption{The construction described in Lemma \ref{lem:sametaxa}. Here $Y$, the
set of taxa missing from at least one tree, is $\{a,c,e,f\}$. In each modified tree the
artificially added members of $Y$ are circled; note that they are always in cherries. A cherry-picking sequence for the original trees is $e, b, c, d, a, f$. A corresponding sequence for the
modified trees is $e, b, c, d, a, f, f', e', c', a', \rho$.}
\label{fig:sameSetOfTaxa}
\end{figure}


First, assume that $\mathcal{T}$ has a cherry-picking sequence $\sigma$. (We show that $\mathcal{T'}$ has a cherry-picking sequence). We start by applying exactly the same sequence of pruning operations to $\mathcal{T'}$. These picking operations will always
be possible because, if a taxon $y \in Y$ is missing from a tree $T_i \in \mathcal{T}$, it will
be in a cherry together with $y'$ in the corresponding tree of $\mathcal{T'}$. After doing this, all the trees will be isomorphic and have the same set of taxa: $Y' \cup \{\rho\}$. At this point these remaining taxa can be pruned in bottom-up fashion (since two isomorphic trees always have a cherry-picking sequence). Hence $\mathcal{T'}$ has a cherry-picking sequence. Note that the taxon $\rho$ is included to ensure that if, during \red{$\sigma$},
a tree $T_i$ has been pruned down to a single taxon, this taxon can still be pruned in the corresponding tree of $\mathcal{T'}$ (because it is sibling to $\rho$).

In the other direction, let $\sigma'$ be a cherry-picking sequence for $\mathcal{T'}$. Let $\sigma$ be the
sequence obtained by deleting all taxa from $\sigma'$ that are not in $X$. Let $x$ be an arbitrary element of $X$ and let $i$ be the position of $x$ in $\sigma'$. Let $\ell'_1, \ell'_2, \ldots, \ell'_{i-1}$ be
the prefix of $\sigma'$ that has been pruned from $\mathcal{T'}$ prior to $x$, and let $\ell_1, \ell_2, \ldots, \ell_j$ (where \red{$j \leq i-1$}) 
be the prefix of $\sigma$ that has been pruned prior to $x$. We claim that, if $x$ is available in 
$\mathcal{T'}[-\ell'_1, \ell'_2, \ldots, \ell'_{i-1}]$, then it is also available in $\mathcal{T}[-\ell_1, \ell_2, \ldots, \ell_{j}]$.
To see this, let $T$ be an arbitrary tree in $\mathcal{T}[-\ell_1, \ell_2, \ldots, \ell_{j}]$. If $x \not \in X(T)$, then (by definition) $x$ is in a cherry of $T$. If $x$ is the only taxon in $T$, then it is (also by definition) in a cherry. So the
only case remaining is that  $x \in X(T)$ and $|X(T)| \geq 2$. Let $T'$ be the tree from $\mathcal{T'}[-\ell'_1, \ell'_2, \ldots, \ell'_{i-1}]$ that corresponds to $T$. The critical observation here is that, by construction, $T$ occurs as a pendant subtree of $T'$. So if
$x$ was not in a cherry of $T$, \revision{then $x$ would not be in a cherry of $T'$ which gives a contradiction to the assumption that $\cT'$ has a cherry-picking sequence}. Hence, $x$ is in a cherry of $T$. Due to the arbitrary choice of $x$ and $T$, it follows that $\sigma$ is a cherry-picking sequence for $\mathcal{T}$.

It remains to show that the reduction
is polynomial time. Observe that,
depending on the instance, the size
of $\mathcal{T}$ can be dominated by
$|X|$ or $m$. Each of the $m$ trees in $\mathcal{T'}$ contains
$|X| + |Y| + 1$ taxa, where $|Y| \leq |X|$, and the transformation itself involves
straightforward operations, so overall the reduction takes poly($|X|$, $m$) time.

\qed
\end{proof} 

Let $\cT$ be a set of rooted binary trees, and let $T_i$ and $T_j$ be two trees in $\cT$ such that $X(T_i) \cap X(T_j) = \emptyset$. Furthermore, let $\rho_i$ and $\rho_j$ be the root vertex of $T_i$ and $T_j$, respectively. Obtain a new tree from $T_i$ and $T_j$ in the following way.
\begin{enumerate}
\item Create a new vertex $\rho$ and add new edges $e=(\rho,\rho_i)$ and  $e'=(\rho,\rho_j)$.
\item Subdivide $e$ (resp. $e'$) with a new vertex $v$ (resp. $v'$) and add a new edge $(v,x)$ (resp. \revision{$(v',y)$}), where $x$ and $y$ are two new taxa such that $\{x,y\}\cap X(\cT)=\emptyset$.
\end{enumerate}
\noindent We call the resulting rooted binary tree the {\it compound tree} of $T_i$ and $T_j$. To illustrate, Figure~\ref{fig:mergeGadget} depicts the compound tree of $T_i$ and $T_j$.

\begin{figure}
\centering
\includegraphics[]{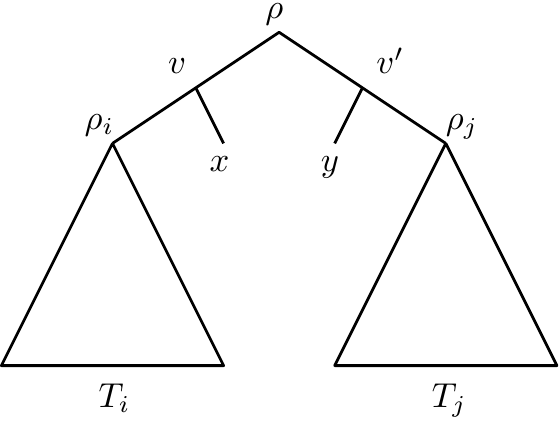}
\caption{The compound tree of two rooted binary trees $T_i$ and $T_j$. The taxon $x$ (resp. $y$) simply ensures that the
last taxon pruned away in the $T_i$ (resp. $T_j$) part \revision{is in a cherry with $x$ (resp. $y$).}}
\label{fig:mergeGadget}
\end{figure}

The next lemma shows that, for a set $\cT$ of rooted binary trees, the replacement of two trees in $\cT$ with their compound tree preserves the  existence and non-existence of a cherry-picking sequences for $\cT$. 

\begin{lemma}
\label{lem:disjointtaxa}
Let $\mathcal{T}$ be a set of rooted binary trees, and let $T_i$ and $T_j$ be two trees in $\mathcal{T}$ such that $X(T_i) \cap X(T_j) = \emptyset$. Let $T_{i,j}$ be the compound tree of $T_i$ and $T_j$. Then $\cT$ has a cherry-picking sequence if and only if $(\cT-\{T_i,T_j\})\cup\{T_{i,j}\}$ has a cherry-picking sequence.
\end{lemma}

\begin{proof} 
To ease reading, let $\cT'=(\cT-\{T_i,T_j\})\cup\{T_{i,j}\}$. Furthermore, let $|X(\cT)|=n$, and let $x$ and $y$ be the unique two taxa in $X(T_{i,j})$ that do not label a leaf in $T_i$ or $T_j$.

Suppose that $\sigma=(\ell_1,\ell_2,\ldots,\ell_n)$ is a cherry-picking sequence for $\cT$. Let $i'$ be the maximum index of an element in $\sigma$ such that $\ell_{i'}\in X(T_i)$ and, similarly, let $j'$ be the maximum index of an element in $\sigma$ such that $\ell_{j'}\in X(T_j)$. Then $\cT[-\ell_1,\ell_2,\ldots\ell_{i'-1}]$ contains a tree that is a single vertex labeled $\ell_{i'}$ and $\cT[-\ell_1,\ell_2,\ldots\ell_{j'-1}]$ contains a tree that is a single vertex labeled $\ell_{j'}$. Moreover, by the construction of $T_{i,j}$, the set $\cT'[-\ell_1,\ell_2,\ldots\ell_{i'-1}]$ contains a tree with cherry $\{\ell_{i'},x\}$ and the set $\cT'[-\ell_1,\ell_2,\ldots\ell_{j'-1}]$ contains a tree with cherry $\{\ell_{j'},y\}$. 
Since $T_i$ and $T_j$ are pendant subtrees in $T_{i,j}$ and $\sigma$ is a cherry-picking sequence for $\cT$, it now follows that $$(\ell_1,\ell_2,\ldots,\ell_n,x,y)$$ is a cherry-picking sequence for $\cT'$.

Conversely, suppose that $\sigma'=(\ell_1,\ell_2,\ldots,\ell_{n+2})$ is a cherry-picking sequence for $\cT'$. Let $\{\ell_{i'},\ell_{j'}\}=\{x,y\}$. Without loss of generality, we may assume that $i'<j'$. Then, as $x$ and $y$ are only contained in the leaf set of $T_{i,j}$, it is straightforward to check that $$(\ell_1,\ell_2,\ldots,\ell_{i'-1},\ell_{i'+1},\ell_{i'+2},\ldots,\ell_{j'-1},\ell_{j'+1},\ell_{j'+2},\ldots,\ell_{n+2})$$ is a cherry-picking sequence for $\cT$.
\qed
\end{proof}

Now we establish the main result of this paper.

\begin{theorem}
\label{thm:hard}
It is {\em NP}-complete to decide if two rooted binary phylogenetic trees $T$ and $T'$ on $X$ have a cherry-picking sequence.
\end{theorem}
\begin{proof} 
Given an order $\sigma=(x_1,x_2,\ldots,x_{|X|})$ on $X$, we can decide in polynomial time if, for each $i\in\{1,2,\ldots,|X|\}$, $x_i$ is in a cherry  in $T[-x_1,x_2,\ldots,x_{i-1}]$ and $T'[-x_1,x_2,\ldots,x_{i-1}]$. Hence, the problem of deciding if $T$ and $T'$ have a cherry-picking sequence is in NP.
To establish the theorem,  we use a reduction from {\sc 3-Sat}. This is the variant of {\sc Satisfiability} where each clause contains \emph{exactly} three literals, and the logical expression is in conjunctive normal form, i.e., 
\begin{equation}
\bigwedge_{i=1}^m C_i = \bigwedge_{i=1}^m (l_{i,1} \vee l_{i,2} \vee l_{i,3}),\notag
\end{equation}
where $l_{i,j} \in \{v^{(k)}, \neg v^{(k)} \mid 1 \leq k \leq n\}$. The corresponding set of variables is denoted with 
\[
V := \{v^{(1)}, v^{(2)}, \ldots, v^{(n)}\}.
\]
We reduce from the NP-complete version of {\sc 3-Sat} in which no variable occurs more than once in a given clause. Such restricted instances can
easily
be obtained by a standard transformation as described in~\cite{garey}. In the remainder of this proof, $n$ and $m$ refer to the number of variables and clauses in a restricted {\sc 3-Sat}  instance, respectively. 


Now, given an instance $I$ of {\sc 3-Sat}, we first construct a set $\cT$ of $3n + 5m + 2$ trees with overlapping taxa sets and  show that $I$ has a satisfying truth assignment if and only if $\cT$ has a cherry-picking sequence. We then repeatedly apply Lemma \ref{lem:disjointtaxa} in order to replace $\cT$ with two trees and, finally, apply Lemma \ref{lem:sametaxa} to complete the proof of this theorem. 

We start by describing the construction of $\cT$ that makes use of the introduction of a set $\{b_1,b_2,\ldots,b_{4n+3m},b_X,b_Y,b_Z\}$ of blocking taxa. As we will see later, each such taxon can only be pruned from $\cT$ after certain other taxa have been pruned first and so the main function of the blocking taxa is to  be unavailable for pruning which in turn constraints the number of possibilities to construct a cherry-picking sequence from $\cT$.
An explicit example of the construction of $\cT$ is given subsequently to this proof. \\
\begin{figure}[t]
\centering
\includegraphics[width=\textwidth]{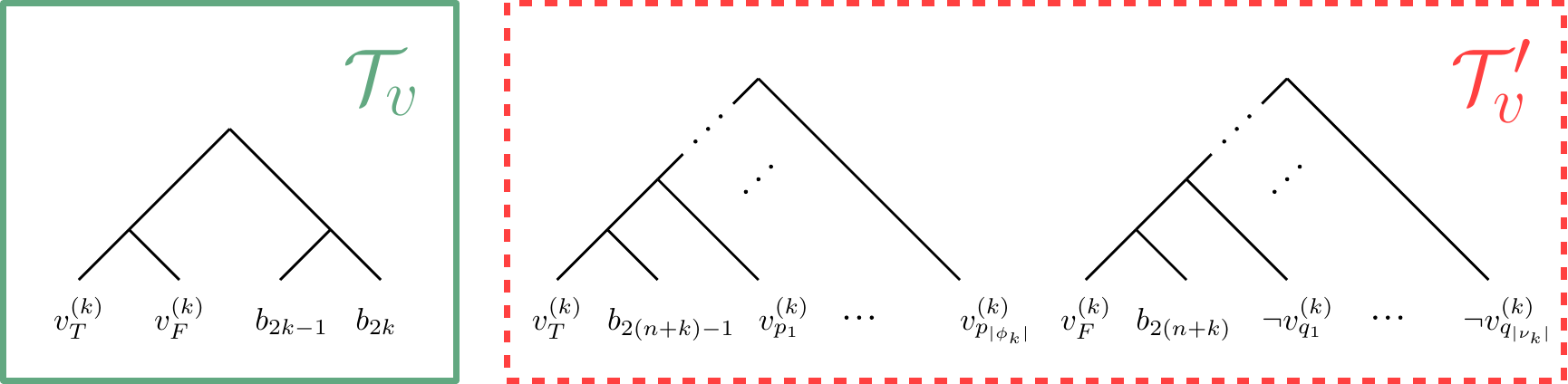}
\caption{Each variable $v^{(k)}$, is represented by a single tree in $\cT_v$ and two trees in $\cT_v'$.}
\label{fig:variableGadget}
\end{figure}

\emph{Variable gadget.} We construct two sets $\cT_v$ and $\cT_v'$ of trees. Each variable $v^{(k)}$ with $k\in\{1,2,\ldots,n\}$ adds  one tree on four taxa to $\cT_v$ which is the tree shown in the solid box of Figure~\ref{fig:variableGadget}. Each such tree has two blocking taxa and, intuitively, encodes whether $v^{(k)}$ is set to be true or false, depending on whether $v_T^{(k)}$ or $v_F^{(k)}$ is pruned first. 
Moreover, each variable $v^{(k)}$ adds two caterpillars to $\cT_v'$. Relative to a fixed $v^{(k)}$, the precise construction of these caterpillars is based on the definition of two particular tuples. Let $\phi_k := (p_1, p_2,\ldots, p_{|\phi_k|})$ (resp. $\nu_k := (q_1, q_2,\ldots, q_{|\nu_k|})$) \revision{be the indices, in ascending order,
of all the clauses in which
$v^{(k)}$ appears unnegated (resp. negated).} Since no clause contains any variable more than once, the elements in $\phi_k$ (resp. $\nu_k$) are pairwise distinct.

Now the taxon set of one caterpillar contains $v_T^{(k)}$, a new blocking taxon and, for each element $p_j$ in $\phi_k$, a new taxon $v_{p_j}^{(k)}$, while the taxon set of the other caterpillar contains $v_F^{(k)}$, a new blocking taxon and, for each element $q_j$ in $\nu_k$, a new taxon $\neg v_{q_j}^{(k)}$. The precise ordering of the leaves in both caterpillars is shown in the dashed box of Figure~\ref{fig:variableGadget}. It is easily checked that $|X(\cT_v)|=4n$ and, since each clause contains precisely three distinct variables, $|X(\cT_v')|=4n+3m$. Noting that the taxa sets of the trees in $X(\cT_v)$ and $X(\cT_v')$ only overlap in $v_T^{(k)}$ and $v_F^{(k)}$, we have 
\begin{equation}\label{eq:variable-count)}
|X(\cT_v\cup\cT_{v}')|=4n+4n+3m-2n=6n+3m
\end{equation}
distinct taxa over all trees in $\cT_v$ and $\cT_{v}'$.\\

\begin{figure}[t]
\centering
\includegraphics[width=.8\textwidth]{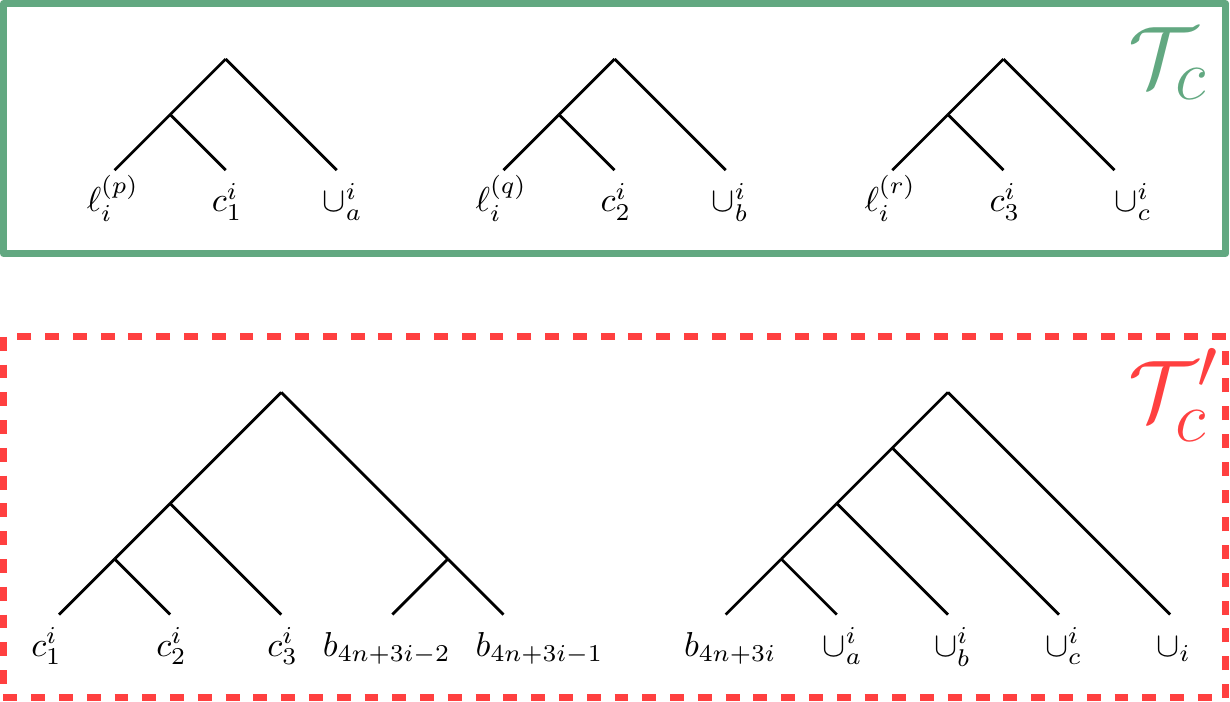}
\caption{Each clause $C_i$ is represented by three trees in $\cT_c$ and two trees in $\cT_c'$.
} 
\label{fig:clauseGadget}
\end{figure}

\emph{Clause gadget.} We construct two sets $\cT_c$ and $\cT_c'$ of trees. For each $i\in\{1,2,\ldots,m\}$, consider the clause $C_i = \ell^{(p)} \vee \ell^{(q)} \vee \ell^{(r)}$, where each  $k\in\{p,q,r\}$ is an element in $\{1,2,\ldots,n\}$ with $\ell^{(k)}\in\{v^{(k)},\neg v^{(k)}\}$. Relative to $C_i$, we add three three-taxon trees to $\cT_c$ which are shown in the solid box of Figure~\ref{fig:clauseGadget}. The first such tree has taxon set $\{\ell_i^{(p)},c_1^i,\cup_a^i\}$ where $\ell_i^{(p)}$ is an element in $\{v_i^{(p)},\neg v_i^{(p)}\}$. Note that $\ell_i^{(p)}$ labels a leaf of a tree in $\cT_v'$ while the other two taxa do not label a leaf of a tree in $\cT_v$ or $\cT_v'$. The other two trees in $\cT_c$ are constructed in an analogous way. Furthermore, for each $C_i$, we add two five-taxon trees to $\cT_c'$ which are shown in the dashed box of Figure~\ref{fig:clauseGadget}. The taxa set of the first tree contains two new blocking taxa and the three previously encountered elements $\{c_1^i,c_2^i,c_3^i\}$, while the second tree contains one new blocking taxon, the new taxon $\cup_i$, and the three previously encountered elements
\revision{$\{\cup_a^i,\cup_b^i,\cup_c^i\}$}. Similar to the variable gadgets, we now count the number of taxa in trees in $\cT_c$ and $\cT_c'$. As no two trees in  $\cT_c$ or $\cT_c'$ share a taxon, we have $|X(\cT_c)|=9m$ and $|X(\cT_c')|=10m$. Moreover, since all taxa of trees in $\cT_c'$, except for the blocking taxa and elements in $\{\cup_1,\cup_2,\ldots,\cup_m\}$, are also taxa of trees in $\cT_c$, we have  
\begin{equation}\label{eq:clause-count)}
|X(\cT_c\cup\cT_c')|=9m+10m-6m=13m.
\end{equation}

\emph{Formula gadget.} We complete the construction of $\cT$ by constructing two caterpillars $T_f$ and $T_f'$ which are shown in the solid and dashed box of Figure~\ref{fig:formulaGadget}, and define $$\cT=\cT_v\cup\cT_v'\cup\cT_c\cup\cT_c'\cup\{T_f,T_f'\}.$$ Summarizing the construction, we have $|\cT|=3n+5m+2$. Moreover, by construction and Equations (\ref{eq:variable-count)})-(\ref{eq:clause-count)}), it follows that $|X((\cT_v\cup\cT_v')\cap(\cT_c\cup\cT_c'))|=3m$. Now, since the three taxa $b_X$, $b_Y$, and $b_Z$, which are common to $T_f$ and $T_f'$, are the only taxa of these two trees that are not contained in the taxa set of any other constructed tree, we have
\begin{equation}\label{eq:noTaxa}
|X(\cT)|=6n+3m+13m-3m+3=6n+13m+3.
\end{equation}


\begin{figure}[t]
\centering
\includegraphics[scale=0.8]{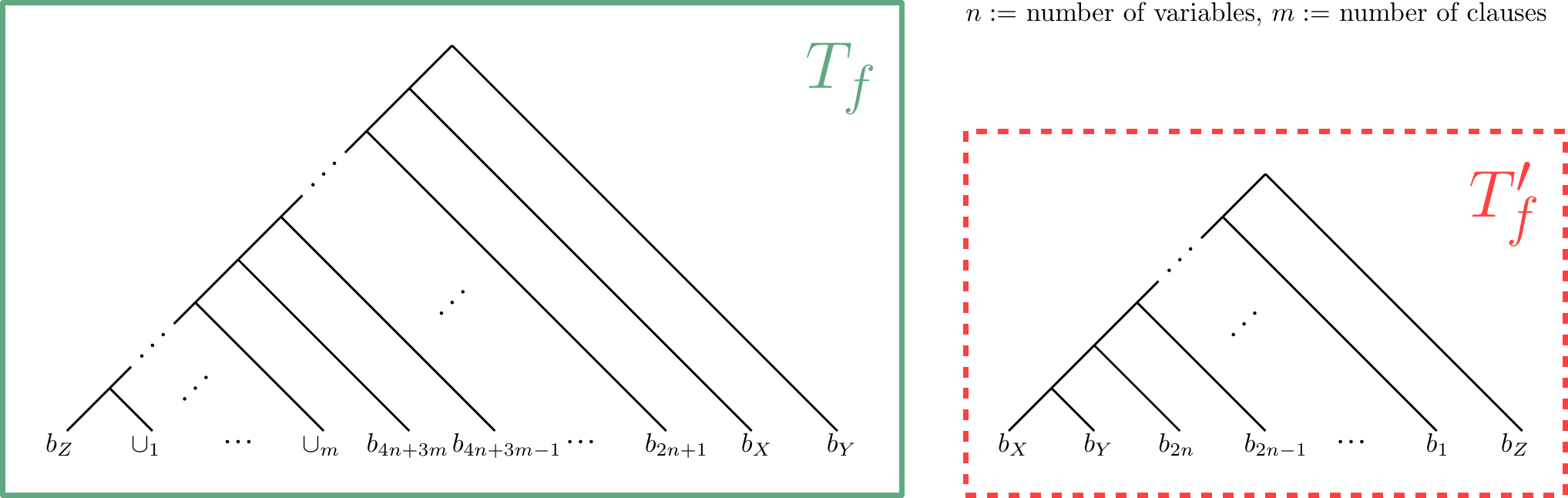}
\caption{The two trees $T_f$ and $T_f'$ in the construction of $\cT$ from $I$.
}
\label{fig:formulaGadget}
\end{figure}
We next prove the following claim:\\

\noindent {\bf Claim 1.}
$I$ is satisfiable if and only if $\cT$ has a cherry-picking sequence.\\

First, suppose that $I$ is satisfiable. Let $\beta: V\rightarrow \{T,F\}$ be a truth assignment for $V$ such that each clause is satisfied. We next describe a sequence of pruning operation. Noting that each taxon in $X(\cT)$ is contained in the taxa sets of exactly two trees in $\cT$ (a fact that we freely use throughout the rest of this proof), it is straightforward to verify that this sequence implies a cherry-picking sequence for $\cT$.

\paragraph*{Part 1: Variable gadgets.}

For each variable $v^{(k)}$ with $k\in\{1,2,\ldots,n\}$ do the following. If $\beta(v^{(k)})=T$
prune taxon $v_T^{(k)}$ from the two trees in $\cT_v\cup\cT_v'$ whose taxa sets contain $v_T^{(k)}$. On the other hand, if $\beta(v^{(k)})=F$
prune taxon $v_F^{(k)}$ from the two trees in $\cT_v\cup\cT_v'$ whose taxa sets contain $v_F^{(k)}$. Taken together, these pruning steps delete a single leaf of each tree in $\cT_v$ and a single leaf of half of the trees in $\cT_v'$. 

\paragraph*{Part 2: Clause gadgets.}
Consider the set of trees resulting from the pruning described in Part 1. For each $C_i= \ell^{(p)} \vee \ell^{(q)} \vee \ell^{(r)}$ with $i\in\{1,2,\ldots,m\}$, let $L_i$ be a subset of $\{p,q,r\}$ such that $|L_i|=2$ and, if $\ell_i^{(k)}$  is not satisfied by $\beta$, then $k\in L_i$. Setting $i=1$, process the three literals in $C_i$ from left to right in the following way.  
\begin{enumerate}
\item If $\ell_i^{(k)}$ is satisfied by $\beta$, prune $\ell_i^{(k)}$ from the tree in $\cT_c$ whose taxa set contains $\ell_i^{(k)}$ and, noting that $\ell_i^{(k)}\in\{v_i^{(k)}, \neg v_i^{(k)}\}$, prune $\ell_i^{(k)}$ from the tree in $\cT_v'$ whose taxa set contains $\ell_i^{(k)}$. 
\item If $k\in L_i$, prune $c_s^i$, where $s=1$ if $k=p$, $s=2$ if $k=q$, and $s=3$ if $k=r$, from the  two trees in $\cT_c\cup\cT_c'$ whose taxa sets contain $c_s^i$.
\item Prune $\cup^{i}_t$, where $t=a$ if $k=p$, $t=b$ if $k=q$, and $t=c$ if $k=r$, from the two trees in $\cT_c\cup\cT_c'$ whose taxa sets contain $\cup^{i}_t$.
\end{enumerate}
Now prune $\cup_i$ from the tree in $\cT_c'$ whose taxa set contains $\cup_i$, and prune $\cup_i$ from $T_f$. If $i<m$, increment $i$ by one and repeat this process with the next clause. Intuitively, by definition of $L_i$, the above process prunes exactly two elements in $\{\magenta{c_1^i},c_2^i,c_3^i\}$. Since each clause is satisfied by $\beta$, this guarantees that we can prune each element in $\{\cup_a^i,\cup_b^i,\cup_c^i\}$ and, subsequently $\cup_i$. 


\paragraph*{Part 3: Formula gadget and remaining taxa.}
Consider the set of trees resulting from the pruning described in Part 2. We prune the remaining taxa as follows.
\begin{enumerate}
\item In order, prune each of
\[
b_{4n+3m},\, b_{4n+3m-1},\, b_{4n+3m-2},\, \ldots, b_{4n+3i},\, b_{4n+3i-1},\, b_{4n+3i-2},\, \ldots,\, b_{4n+3},\, b_{4n+2},\, b_{4n+1}
\]
from a tree $\cT_c'$ whose taxa set contains the respective blocking taxa and from $T_f$. After all taxa have been pruned, each tree in $\cT_c'$ is either the null tree or consists of a single vertex labeled $c_s^i$ for some $s\in\{1,2,3\}$.
\item For each $i\in\{1,2,\ldots,m\}$, prune the unique taxon $c_s^i$ with $s\in\{1,2,3\}$ that has not been pruned in Part 2 from two trees in $\cT_c\cup\cT_c'$. Now, each tree in $\cT_c\cup\cT_c'$ that is not the null tree consists of a single vertex labeled $\ell_i^{(k)}$ for some $i\in\{1,2,\ldots,m\}$ and $k\in\{1,2,\ldots,n\}$.
\item For each $k\in\{1,2,\ldots,n\}$, note that one of $\{b_{2(n+k)-1},b_{2(n+k)}\}$ labels a leaf of a cherry in a tree in $\cT_v'$ while the other labels the leaf of a tree in $\cT_v'$ that consists of a single vertex. In order, prune each of
\[
b_{4n},\, b_{4n-1},\, \ldots, b_{2(n+k)},\, b_{2(n+k)-1},\, \ldots,\, b_{2n+2},\, b_{2n+1}
\]
from the  tree in $\cT_v'$ whose taxa set contains the respective blocking taxa and from $T_f$. 
\item In order, prune $b_X$ and $b_Y$ from $T_f$ and $T_f'$. 
\item Consider the remaining trees in $\cT_v$ and observe that each such tree consists of exactly three leaves, two of which are blocking taxa that form a cherry. In order, prune each of
\[
b_{2n},\, b_{2n-1},\, \ldots, b_{2k},\, b_{2k-1},\, \ldots,\, b_{2},\, b_{1}
\]
from $T_f'$ and the tree in $\cT_v$ whose taxa set contains the respective blocking taxon. 
\item For each $k\in\{1,2,\ldots,n\}$, let $v_X^{(k)}$ be the unique element in $\{v_T^{(k)},v_F^{(k)}\}$ that has not been pruned in Part 1. Prune $v_X^{(k)}$ from the two trees in $\cT_v\cup\cT_v'$ whose taxa sets contain $v_X^{(k)}$.
\item For each $i\in\{1,2,\ldots,m\}$ in increasing order, consider each literal $\ell_i^{(k)}$ in $C_i= \ell^{(p)} \vee \ell^{(q)} \vee \ell^{(r)}$ with $k\in\{p,q,r\}$ that is not satisfied by $\beta$. By processing such literals from left to right in $C_i$,  prune $\ell_i^{(k)}$ from the two trees in $\cT_v'\cup\cT_c$ whose taxa sets contain $\ell_i^{(k)}$. It is easily seen that the corresponding tree in $\cT_v'$ either consists of a single vertex or \magenta{contains} a cherry with a leaf labeled $\ell_i^{(k)}$.
\item Prune $b_Z$ from $T_f$ and $T_f'$.
\end{enumerate}

Now, relative to the elements in $X(\cT)$, we prune $2n$ elements in Parts 1 and 3.6, all $4m$ elements in $$\{\cup_1,\cup_a^i,\cup_b^i,\cup_c^i,\ldots,\cup_m,\cup_a^m,\cup_b^m,\cup_c^m\}$$ in Part 2, and all $4n+3m+3$ blocking taxa in Parts 3.1, 3.3, 3.4, 3.5, and 3.8.
Additionally, in Parts 2.1 and 3.7 we prune $3m$ taxa, and in Parts 2.2 and 3.2, we prune again $3m$ taxa. Summing up, we prune $$6n+13m+3$$ taxa, which is equal to the number of elements in $X(\cT)$.

Second, suppose that $\cT$ has a cherry-picking sequence $\blue{\sigma} = (x_1,x_2,\ldots, x_{|\blue{\sigma}|})$. We write $x_i \prec x_j$ if and only if $i < j$ and $x_i \succ x_j$ if and only if $i > j$. Further, let 
\[
M := \{1,2,\ldots,m\}, \; N := \{1,2,\ldots, n\}, \; B := \{b_1,b_2,\ldots,b_{4n+3m},b_X,b_Y,b_Z\}.
\]
We define a truth assignment $\beta\colon V \rightarrow \{T,F\}$ as follows
\[
\beta\left (v^{(k)}\right ) = 
\begin{cases}
T & \text{if } \exists i \in M\colon v_T^{(k)} \prec \cup_i, \\
F & \text{else.}
\end{cases}
\]
In order to show that $\beta$ \revision{satisfies} each clause \blue{of $I$}, we \blue{establish four}  necessary conditions that \blue{$\sigma$ fulfills} by construction. 
\begin{enumerate}
\item All taxa in $\{\cup_1, \cup_2, \ldots, \cup_m\}$ are pruned earlier than any blocking taxon:
\begin{equation}\label{eq:cond-one}
\forall i \in M \; \forall b \in B\colon \cup_i \prec b. 
\end{equation}
\magenta{\emph{Argument: }}Observe that the arrangement of \revision{$b_X$, $b_Y$, $b_Z$} \blue{in $T_f$ and $T'_f$ implies that all taxa in $\{\cup_1, \cup_2, \ldots, \cup_m\}$ are pruned prior to any blocking taxon.} 
Furthermore, we cannot prune any taxon in $T'_f$ until we \blue{have pruned all taxa from $T_f$ except for $b_X$, $b_Y$, and $b_Z$.} 
\blue{We will freely use Condition \magenta{1} throughout the remainder of this proof.}

\item Let \revision{$C_i = \ell^{(p)} \vee \ell^{(q)} \vee \ell^{(r)}$} be a clause of $I$. At least one taxon in $\{\ell_i^{(p)}, \ell_i^{(q)}, \ell_i^{(r)}\}$ is pruned earlier than $\cup_i$. Stated more formally: 
\begin{equation}
\forall i \in M \; \exists \ell_i^{(s_i)} \in  \left \{\ell_i^{(p_i)}, \ell_i^{(q_i)}, \ell_i^{(r_i)}\right \}\colon \ell_i^{(s_i)} \prec \cup_i.
\end{equation}
\magenta{\emph{Argument: }}Consider the five trees in $\cT_c \cup \cT'_c$ representing \blue{$C_i$} (see Figure~\ref{fig:clauseGadget}). In order to prune $\cup_i$, we have to prune all taxa in $\{\cup_a^i,\cup_b^i,\cup_c^i\}$ first. Since we can prune at most two taxa in $\{c^{i}_1, c^{i}_2, c^{i}_3\}$ \blue{prior to an element in $\{b_{4n+3i-2},b_{4n+3i-1}\}$}, pruning all taxa in $\{\cup_a^i,\cup_b^i,\cup_c^i\}$ is only possible if at least one   taxon in $\{\ell_i^{(p)}, \ell_i^{(q)}, \ell_i^{(r)}\}$ has been pruned previously. 

\item Let $v^{(k)} \in V$ be any variable of $I$. \blue{Recall the definition of the tuples $\phi_k$ and $\nu_k$ that is used in the construction of the variable gadget. If there exists a  $v_{i}^{(k)}$ with $v_{i}^{(k)} \prec \cup_i$ for some $i \in \phi_k$, then $v_T^{(k)}$ is also pruned earlier than $\cup_i$. Stated formally:} 
\begin{equation}
\forall k \in N \; \forall i \in \phi_k \colon \left ( v_{i}^{(k)} \prec \cup_{i} \implies  v_T^{(k)} \prec \cup_i \right ).
\end{equation}
\magenta{\emph{Argument: }}Consider a variable $v^{(k)} \in V$ such that $v_{i}^{(k)} \prec \cup_i$ for some $i \in \phi_k$. Since there is no blocking taxon $b \in B$ with $b \prec \cup_i$, we have \blue{$\cup_i \prec b_{2(n+k)-1}$}. Thus, $v_T^{(k)}$  is pruned from the \blue{associated caterpillar in $\cT'_v$ that contains} $v_{i}^{(k)}$  such that $v_T^{(k)} \prec v_{i}^{(k)} \prec \cup_i$ (see Figure~\ref{fig:variableGadget}). 

The following can be shown analogously. \blue{If there exists a $\neg v_{i}^{(k)} \prec \cup_i$ for some $i \in \nu_k$, then $v_F^{(k)}$ is also pruned earlier than $\cup_i$. Stated formally:}
\begin{equation}
\forall k \in N \; \forall i \in \nu_k \colon \left (\neg v_{i}^{(k)} \prec \cup_{i} \implies  v_F^{(k)} \prec \cup_i \right ).
\end{equation}

\item Let $v^{(k)} \in V$ be any variable of $I$. If $v_T^{(k)}$ is pruned earlier than some taxon in $\{\cup_1, \cup_2,\ldots, \cup_m\}$, then $v_F^{(k)}$ is pruned later than all taxa in $\{\cup_1, \cup_2,\ldots, \cup_m\}$, i.e., 
\begin{equation}\label{eq:v_T_early}
\forall k \in N \colon \left ( \left ( \exists i \in M \colon v_T^{(k)} \prec \cup_i \right ) \implies \left ( \forall i \in M \colon v_F^{(k)} \succ \cup_i \right ) \right ).
\end{equation}
\magenta{\emph{Argument: }}Consider a variable $v^{(k)} \in V$ such that $v_T^{(k)} \prec \cup_i$ for some $i \in M$. Assume towards a contradiction that there is \blue{some} $j \in M$ such that $v_F^{(k)} \prec \cup_j$. Then, one of the \blue{two blocking taxa $b_{2k-1}$ and $b_{2k}$ is pruned prior to $v_F^{(k)}$} (see Figure~\ref{fig:variableGadget}). But this is not possible since there is no blocking taxon $b \in B$ with $b \prec \cup_j$. 

As an immediate consequence of statement~\eqref{eq:v_T_early}, we get the analogous statement for $v_F^{(k)}$, i.e., 
\begin{equation}
\forall k \in N \colon \left ( \left ( \exists i \in M \colon v_F^{(k)} \prec \cup_i \right ) \implies \left ( \forall i \in M \colon v_T^{(k)} \succ \cup_i \right ) \right ).
\end{equation}

\end{enumerate}

Now, we show that $\beta$ indeed satisfies each clause of $I$. For each clause \revision{$C_i = \ell^{(p)} \vee \ell^{(q)} \vee \ell^{(r)}$}, we have $\ell_i^{(s)} \prec \cup_i$ for some $\ell_i^{(s)} \in  \left \{\ell_i^{(p)}, \ell_i^{(q)}, \ell_i^{(r)}\right \}$ (Condition 2). Since $\ell_i^{(s_i)}\in \left  \{v_i^{(k)}, \neg v_i^{(k)} \right \}$ for some $k \in N$, we have $v_T^{(k)} \prec \cup_i$ if $\ell_i^{(s_i)} = v_i^{(k)}$ and $v_F^{(k)} \prec \cup_i$ if $\ell_i^{(s_i)} = \neg v_i^{(k)}$ (Condition 3). Hence, by setting $\beta(v^{(k)}) = T$ if $v_T^{(k)} \prec \cup_i$ and $\beta(v^{(k)}) = F$ if $v_F^{(k)} \prec \cup_i$, we satisfy at least one literal of each clause. Note that we can assign arbitrary truth values to variables $v^{(k)}$ with $v_T^{(k)} \succ \cup_i$ and $v_F^{(k)} \succ \cup_i$ for all $i \in M$. Here, we choose to set all these variables to $F$. The truth assignment $\beta$ is consistent, since at least one taxon in $\left \{ v_T^{(k)}, v_F^{(k)} \right \}$ is pruned later than all taxa in $\{\cup_1, \cup_2, \ldots, \cup_m\}$ (Condition 4). Hence, the truth assignment $\beta$ is consistent and satisfies each clause of $I$.

\paragraph*{Folding into two trees on the same set of taxa.} The trees in $\cT_v\cup\cT_c\cup\{T_f\}$ and, similarly, the trees in $\cT_v'\cup\cT_c'\cup\{T_f'\}$  (see Figures~\ref{fig:variableGadget},~\ref{fig:clauseGadget}, and~\ref{fig:formulaGadget})  have mutually disjoint taxa sets. Hence, by $n+3m$ applications of Lemma~\ref{lem:disjointtaxa}, we can construct a compound tree $S$ for all trees in $\cT_v\cup\cT_c\cup\{T_f\}$ and, by $2n+2m$ applications of Lemma~\ref{lem:disjointtaxa}, we can construct a compound tree $S'$ for all trees in $\cT_v'\cup\cT_c'\cup\{T_f'\}$ such that $\cT$ has a cherry-picking sequence if and only if $S$ and $S'$ have a cherry-picking sequence. 
Lastly, by appying Lemma~\ref{lem:sametaxa}, we obtain two trees $T$ and $T'$ from $S$ and $S'$, respectively, such that $X(T)=X(T')$, and $S$ and $S'$ have a cherry-picking sequence if and only if $T$ and $T'$ have such a sequence. It now follows that $I$ is satisfiable if and only if $T$ and $T'$ have a cherry-picking sequence.

\paragraph*{Number of taxa in the final instance.} It remains to show that $T$ and $T'$ can be constructed in polynomial time. By Equation~\ref{eq:noTaxa}, recall that $|X(\cT)|=6n+13m+3$. Now, since we apply Lemma~\ref{lem:disjointtaxa} a total of $3n+5m$ times and each application introduces two new taxa, we have $$\magenta{|X(\{S,\,S'\})|}= 6n+13m+3+ 2(3n+5m)=12n+23m+3.$$
Observe that each taxon in $X(\cT)$ labels a leaf of a unique tree in $\cT_v\cup\cT_c\cup\{T_f\}$ and a leaf of a unique tree in $\cT_v'\cup\cT_c'\cup\{T_f'\}$. It therefore follows that each taxon that is contained in exactly one of $X(S)$ and $X(S')$ has been introduced by an application of Lemma~\ref{lem:disjointtaxa}. Conversely, each application of this lemma introduces two taxa that are both contained in exactly one of $X(S)$ and $X(S')$.
Hence, \revision{recalling that in obtaining $T$ and $T'$ from $S$ and $S'$, respectively, an additional leaf labeled $\rho$ is introduced (see the third sentence in the proof of Lemma~\ref{lem:sametaxa}), we have}
$$|X(T)|=|X(T')| = 12n+23m+3+2(3n+5m)+1=18n+33m+4.$$
It now follows, that the size of $T$ and $T'$ as well as the time it takes to construct these two trees are polynomial. This completes the proof of the theorem.\qed
\end{proof}

\begin{figure}[t]
\centering
\includegraphics[scale=0.8]{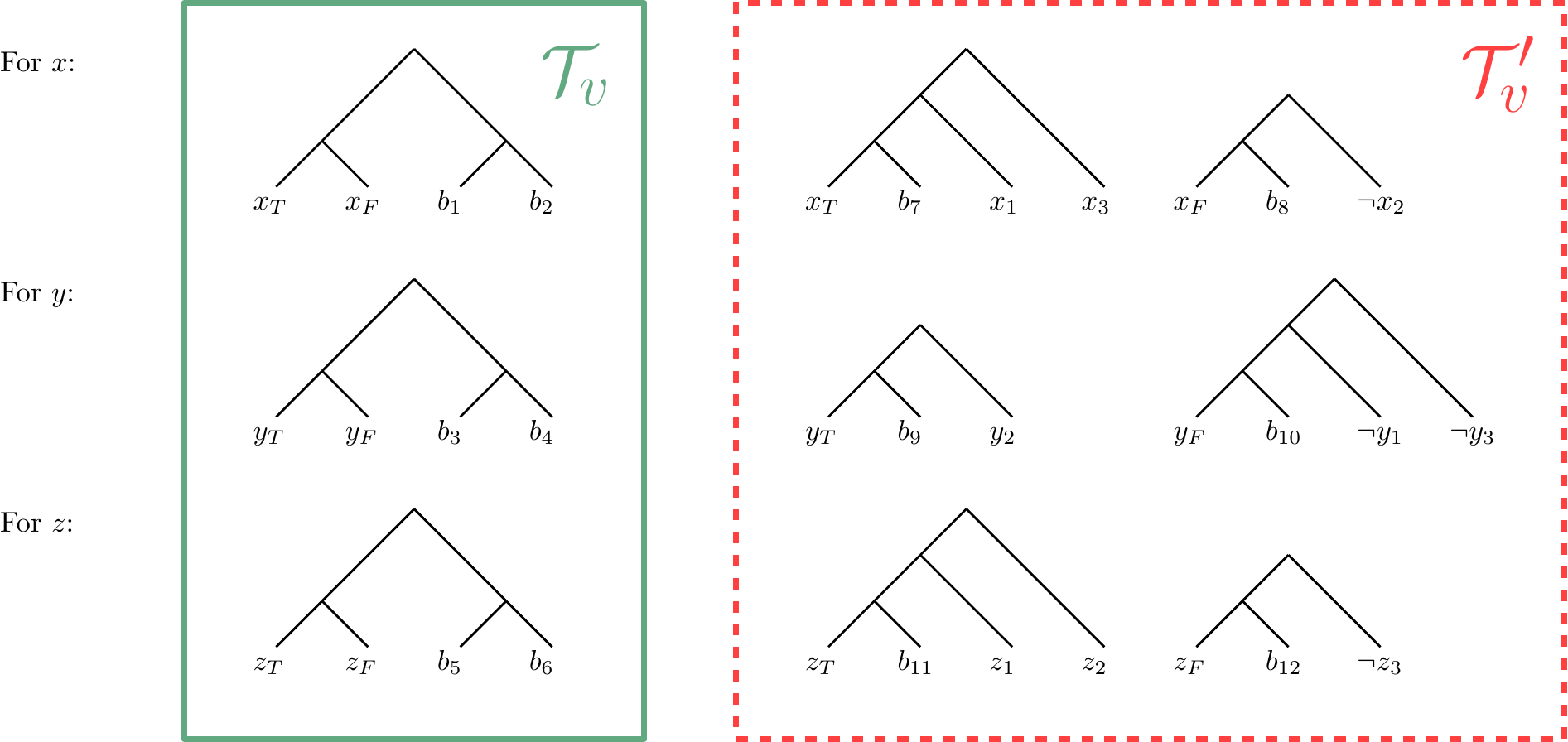}
\caption{The variable gadget for $(x \vee \neg{y} \vee z) \wedge (\neg x \vee y \vee z) \wedge (x \vee \neg y \vee \neg z)$.
}
\label{fig:variableGadgetsExample}
\end{figure}

\begin{figure}[t]
\centering
\includegraphics[scale=0.8]{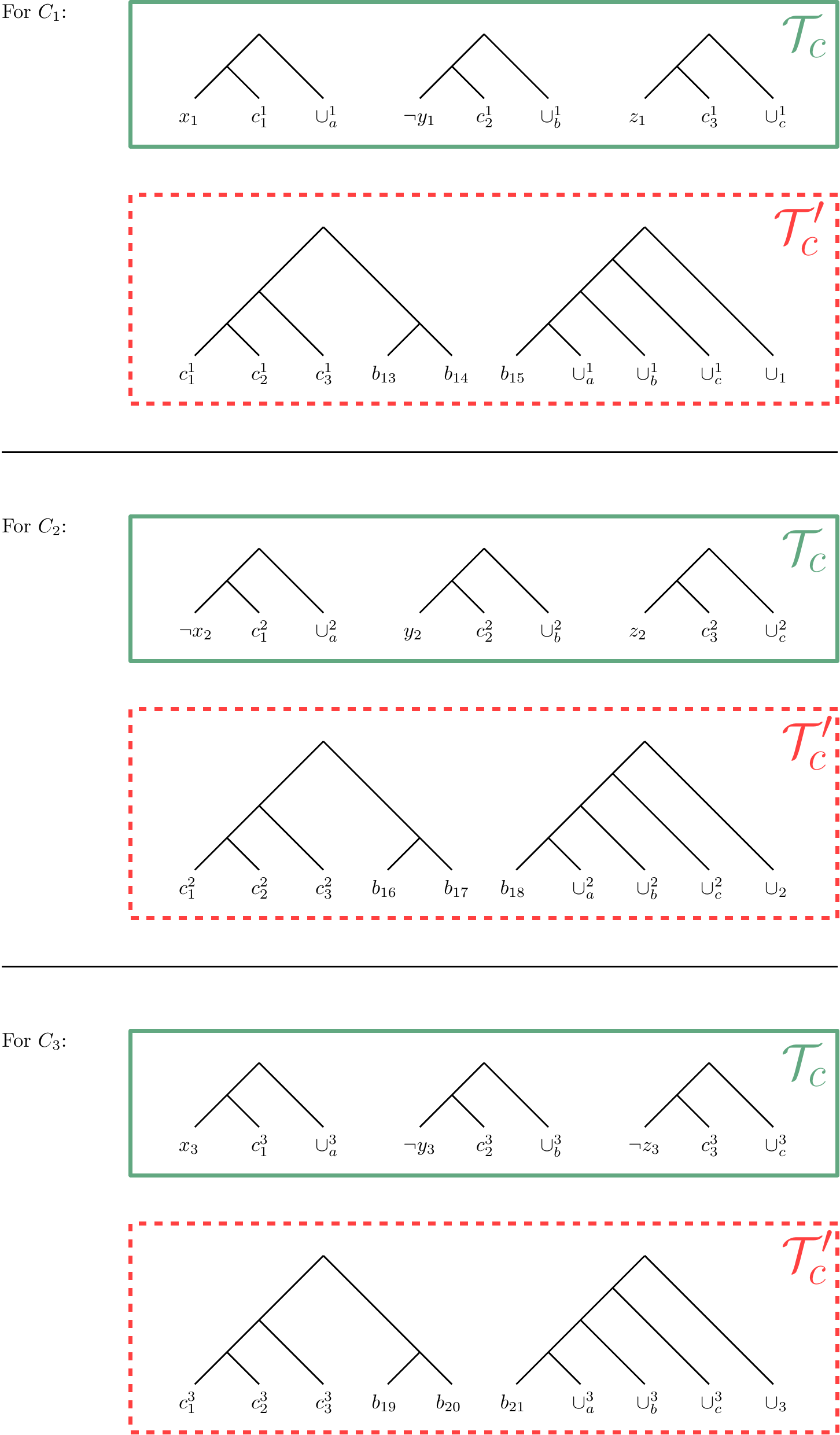}
\caption{The clause gadget for $(x \vee \neg{y} \vee z) \wedge (\neg x \vee y \vee z) \wedge (x \vee \neg y \vee \neg z)$
}
\label{fig:clauseGadgetsExample}
\end{figure}

\begin{figure}[t]
\centering
\includegraphics[scale=0.8]{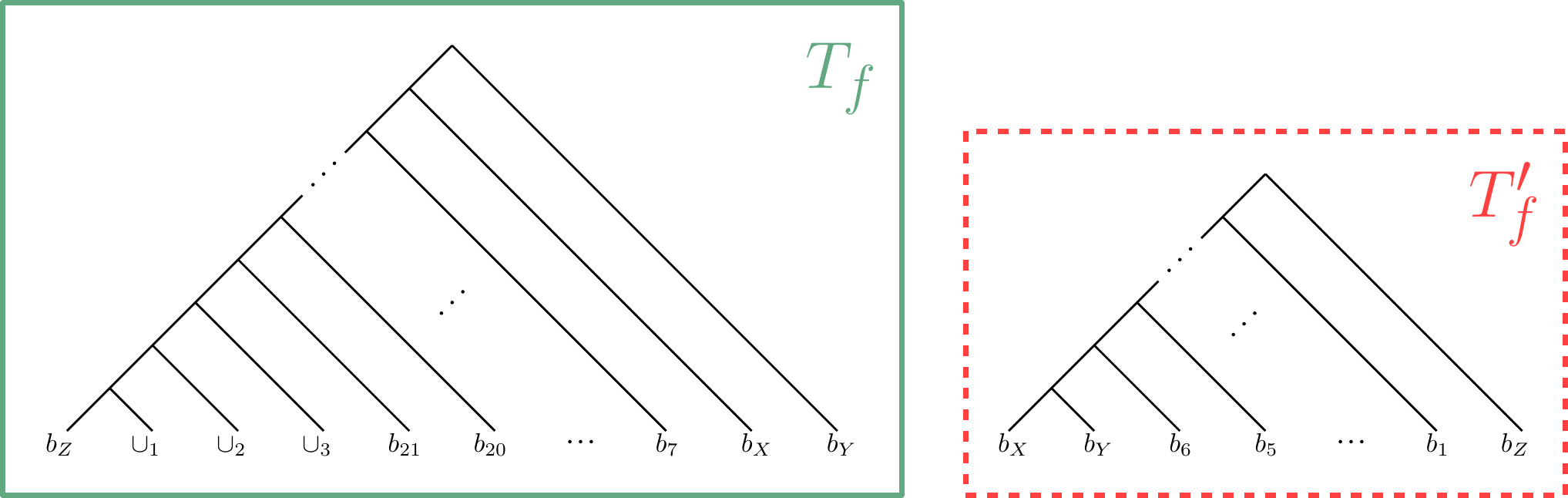}
\caption{The formula gadget for $(x \vee \neg{y} \vee z) \wedge (\neg x \vee y \vee z) \wedge (x \vee \neg y \vee \neg z)$}
\label{fig:formulaGadgetsExample}
\end{figure}

To illustrate the proof of Theorem~\ref{thm:hard}, we now give an explicit example of a {\sc 3-Sat} instance and show how it is reduced to a set of trees by following the construction that is described in the aforementioned proof. Let $I$ be the following instance of {\sc 3-Sat} 
\[
\underbrace{(x \vee \neg{y} \vee z)}_{C_1} \wedge \underbrace{(\neg x \vee y \vee z)}_{C_2} \wedge \underbrace{(x \vee \neg y \vee \neg z)}_{C_3}.
\]
For the purpose of ordering the blocking taxa in the same way as described in the proof, we regard variable $x$ as $v^{(1)}$, variable $y$ as $v^{(2)}$, and variable $z$ as $v^{(3)}$.
Let $n=3$ (resp. $m=3)$ be the number of variables (resp. clauses) in $I$.
We construct a set $\cT$ of $3n+5m+2=26$ trees. The 9 trees that represent the variable gadget $\cT_v$ and $\cT_v'$ are shown in Figure~\ref{fig:variableGadgetsExample}, the 15 trees that represent the clause gadget $\cT_c$ and $\cT_c'$ are shown in Figure~\ref{fig:clauseGadgetsExample} and the two trees that represent the formula gadget $T_f$ and $T_f'$ are shown in Figure~\ref{fig:formulaGadgetsExample}. Note that $|X(\cT)|=3n+13m+3=60$. Clearly, $I$ is satisfied for the truth assignment $\beta: \{x,y,z\}\rightarrow \{T,F\}$ with $\beta(x)= \beta(z)= T, \beta(y) = F$. To see that $\cT$ also has a cherry-picking sequence of length 60, we follow the sequence of pruning operations that is described in Parts 1-3 in the first direction of the proof of Claim 1
\begin{align*}
(&x_T, y_F, z_T, \\
&x_1, c_1^1, \cup_a^1, \neg{y_1}, c_2^1, \cup_b^1, z_1, \cup_c^1, \cup_1, c_1^2, \cup_a^2, c_2^2, \cup_b^2, z_2, \cup_c^2, \cup_2, x_3, \cup_a^3, \neg y_3, c_2^3, \cup_b^3, c_3^3, \cup_c^3, \cup_3, \\
&b_{21}, \blue{b_{20}}, \ldots, b_{13},c_3^1, c_3^2, c_1^3,\blue{b_{12}, b_{11}}, \ldots, b_{7},\revision{b_X}, \revision{b_Y},b_{6}, b_{5}, \ldots, b_{1},x_F, y_T, z_F,\neg x_2, y_2, \neg z_3,\revision{b_Z}),
\end{align*}
where line 1 corresponds to Part 1, line 2 corresponds to Part 2, and line 3 corresponds to Part 3.

\section{Discussion}

Given any set of input trees, there always exists some phylogenetic network displaying them. Roughly speaking, one can simply merge the input trees at the leaves \revision{and at the root}. However, what happens when you restrict the network to have some additional, biologically motivated, properties? Then there might not always exist a network displaying the input trees. Moreover, deciding whether or not there exists such a network may be a difficult problem. Indeed, in this paper we have shown that even if the input consists of only two binary trees, it is already \red{NP-complete} to decide whether there exists any temporal phylogenetic network displaying them.

One could be tempted to look for approximation algorithms for the associated optimization problem: given a set of phylogenetic trees, find a temporal network that displays them and has smallest possible reticulation number, if such a network exists. Note, however, that an approximation algorithm is required to always output a valid solution, for any valid input. The problem formulation above (based on~\cite{humphries2013complexity}) does not specify what a valid solution is when there does not exist a temporal network displaying the input trees. Nevertheless, whatever the output in that case \red{is},
it can be checked in polynomial time whether the output of the algorithm is a temporal network displaying the input trees. \red{This is because temporal networks are tree-child, and checking whether a tree-child network displays a tree can be achieved in polynomial time \cite{locating}}. Hence, any approximation algorithm for the problem could be used to decide in polynomial time whether there exists a temporal network displaying the input trees, which is not possible, unless P$=$NP, given the \red{NP-completeness} shown in this paper.

Therefore, a more promising direction is to consider fixed-parameter algorithms for the associated parameterized version of the problem. Given a set of phylogenetic trees and a parameter~$k$, decide whether there exists a temporal network that displays the input trees and has reticulation number at most~$k$. One then aims at algorithms solving this problem in $O(|X|^{O(1)}f(k))$ time, with~$f$ some function of~$k$, preferably of the form~$c^k$ with~$c$ a small constant. \red{Intuitively}, such an FPT algorithm is only exponential in the reticulation number and not in the number of leaves. Indeed, even though it is \red{NP-complete} to decide whether there exists a temporal network with unlimited reticulation number, for small reticulation numbers this problem might be much easier. In fact, for instances of two binary trees a fixed-parameter algorithm is already known~\cite{humphries2013cherry}. Important open problems include the question whether such algorithms exist for instances of more than two trees and whether algorithms can be developed that work well in practice.

It would also be interesting to consider other biologically motivated network classes. For example, binary tree-child (e.g.~\cite{tree-child}) or tree-sibling networks (e.g.~\cite{tree-sibling}). Could it be that one of the associated decision problems is nontrivial (for more than two input trees) but polynomial-time solvable? For other network classes, such as tree-based (e.g.~\cite{tree-based}) or time-consistent (e.g.~\cite{time-consistent}) networks, it is known that there always exists a solution~\cite{universal-tree-based}. For such classes, it would be interesting to study the optimization version of the problem.

%

\section{Acknowledgements}
We thank Mat\'{u}\v{s} Mihal\'{a}k for useful discussions. Leo van Iersel was partly supported by the Netherlands Organization for Scientific Research (NWO), including Vidi grant 639.072.602, and partly by the 4TU Applied Mathematics Institute. Simone Linz was supported by the New Zealand Marsden Fund. 

\bibliography{bibliographyCherry}{}
\bibliographystyle{plain}

\end{document}